\documentclass{llncs}

\usepackage{amsmath}
\usepackage[T1]{fontenc}
\usepackage{amsfonts}
\usepackage{amssymb}
\usepackage{float}
\usepackage{url}
\usepackage{tikz}
\usetikzlibrary{automata}
\usetikzlibrary{arrows}
\usetikzlibrary{shapes}
\usetikzlibrary{decorations.pathmorphing}
\usepackage{algorithm}
\usepackage{algorithmic}
\usepackage[margin=1in]{geometry}

\tikzstyle{every picture}=[
  >=stealth', shorten >=1pt, node distance=1.44cm,auto,bend angle=45,initial text=,
  every state/.style={inner sep=0.75mm, minimum size=1mm},font=\scriptsize,
]

\newcommand{\cqfd}{\hfill{\vrule height 3pt width 5pt depth 2pt}}

\begin{document} 

  \title{Root-Weighted Tree Automata and their Applications to Tree Kernels}
  
  \author{
    Ludovic Mignot,
    Nadia Ouali-Sebti,
    Djelloul Ziadi
  } 

  \institute{
    LITIS, Universit\'e de Rouen, 76801 Saint-\'Etienne du Rouvray Cedex, France\\
     \email{\{ludovic.mignot, nadia.ouali-sebti, djelloul.ziadi\}@univ-rouen.fr}\\
  }
  
  \maketitle
  
  
  \begin{abstract} 
    In this paper, we define a new kind of weighted tree automata where the weights are only supported by final states.
    We show that these automata are sequentializable and we study their closures under classical regular and algebraic operations.
    
    We then use these automata to compute the subtree kernel of two finite tree languages in an efficient way.
    Finally, we present some perspectives involving the root-weighted tree automata.
  \end{abstract}

\section{Introduction}\label{se:int}

  Kernel methods have been widely used to extend the applicability of many well-known algorithms, such as the Perceptron~\cite{aizerman65}, Support Vector Machines~\cite{cortes95}, or Principal Component Analysis~\cite{zelenkoAR03}
  .
  Tree kernels are interesting 
  approaches
  in areas of machine learning based natural language processing. 
  They have been applied to reduce such effort for several natural language tasks, 
  \emph{e.g.}
  relation extraction~\cite{icann97}, syntactic parsing re-ranking~\cite{CollinsDuffy}, named entity recognition\cite{culottaS04,cumbyR03} and Semantic Role Labeling~\cite{moschitti04}.

  The main idea of tree kernels is to compute the number of common substructures (subtrees and subset trees) between two trees $t_1$ and $t_2$. 
  In~\cite{moschitti06}, Moschitti defined an algorithm for the computation of this type 
  of
  tree kernels which computes the kernels between two syntactic parse trees in $O(m \times n)$ time, where $m$ and $n$ are the number of nodes in the two trees.
  To do this, Moschitti modified the function proposed by Collins and Duffy in~\cite{CollinsDuffy} 
  by
  introducing a parameter $\sigma\in\{0,1\}$ which enables the SubTrees ($\sigma=1$) or the SubSet Trees ($\sigma=0$) evaluation and which is defined for two 
  trees
  $t_1$ and $t_2$ as follows: 
  Given a set of substructures ${\cal S}=\{s_1,s_2,\ldots\}$, they defined the indicator function $I_i(n)$ which is equal 
  to
  $1$ if the substructure $s_i$ is rooted at node $n$ and $0$ otherwise.
  They defined the tree kernel function between the two trees $t_1$ and $t_2$ as follows: 
  $K(t_1,t_2)=\sum_{n_{1} \in N_{t_1}}\sum_{n_{2} \in N_{t_2}}\Delta(n_{1},n_{2})$ where $N_{t_1}$ and $N_{t_2}$ are the number of nodes in $t_1$ and $t_2$ respectively and $\Delta (n_{1},n_{2}) = \sum_{i=1}^{|{\cal S}|} I_{i}(n_{1})\cdot I_{i}(n_{2})$. 
  We can then compute $\Delta$ as follows:
  \begin{itemize}
    \item if the productions at $n_{1}$ and $n_{2}$ are 
    different
    then $\Delta (n_{1},n_{2}) = 0$, 
    \item if the productions at $n_{1}$ and  $n_{2}$ are the same and $n_{1}$ and $n_{2}$ are leaves then $\Delta (n_{1},n_{2}) = 1$,
    \item if the productions at $n_{1}$ and  $n_{2}$ are the same and $n_{1}$ and $n_{2}$ are not leaves then $\Delta (n_{1},n_{2}) = \prod_{j=1}^{nc(n_{1})}(\sigma+  \Delta (C_{n_{1}}^{j}, C_{n_{2}}^{j}))$,  where $nc(n_1)$ is the number of children of $n_1$ and $C_{n}^{j}$ is $j^{th}$ child of the node $n$. 
  \end{itemize}
  
  In~\cite{moschitti2006}, Moschitti proposed a new convolution kernel, namely the Partial Tree kernel, to fully exploit dependency trees. 
  He proposed an efficient algorithm for its computation which is based on  applying the selection of tree nodes with non-null kernel.
  In the following we propose a new 
  technique
  to compute these 
  kind
  of tree kernels using weighted tree automata. 
  We will 
  start
  by defining a new class of weighted tree automata that we call 
  \emph{rooted weighted tree automata}
  and we will prove some properties of 
  these
  weighted tree automata. 
  Then we will show that tree kernels can be computed efficiently using a general intersection of rooted weighted tree automata defined here. 
  
  The paper is organized as follows. 
  In the following section, we introduce the trees, operations in trees and in tree languages and some preliminary notions used in the remaining sections. 
  Section~\ref{sec tree ser rwta}, 
  presents 
  the
  sequentialization of rooted weighted tree automata and the closure of these automata under 
  rational
  or algebraic operations.
  In Section~\ref{sec subtree Ker} we present an efficient computation of the subtree kernel of two finite tree series. 
  Finally, the different results described in this paper are given in the conclusion.  

\section{Preliminaries}\label{sec prelim} 

  Let $\Sigma$ be a graded alphabet.
  A \emph{tree} $t$ over $\Sigma$ is inductively defined $t=f(t_1,\ldots,t_k)$ where $k$ is any integer, $f$ is any symbol in $\Sigma_k$ and $t_1,\ldots,t_k$ are any $k$ trees over $\Sigma$. 
  We denote by $T_{\Sigma}$ the set of trees over $\Sigma$.
  A \emph{tree language} over $\Sigma$ is a subset of $T_\Sigma$.
  
  Let $c$ be a symbol in $\Sigma_0$, $L$ be a tree language over $\Sigma$ and $t$ be a tree in $T_\Sigma$.
  The \emph{tree substitution} of $c$ by $L $ in $t$, denoted by $t_{\{c \leftarrow L\}}$, is the language inductively defined by:
  \begin{itemize}
    \item $L$ if $t=c$;
    \item $\{d\}$ if $t=d \in \Sigma_0\setminus\{c\}$;
    \item $f({t_1}_{\{c \leftarrow L\}}, \ldots, {t_k}_{\{c \leftarrow L\}})$ if $t=f(t_1, \ldots, t_k)$ with $f\in\Sigma_k$ and $t_1, \dots, t_k$ any $k$ trees over $\Sigma$.
  \end{itemize}
  
  \noindent The $c$-\emph{product} $L_1\cdot_{c} L_2$ of two tree languages $L_1$ and $L_2$ over $\Sigma$ is the tree language $L_1\cdot_{c} L_2$ defined by $\bigcup_{t\in L_1}t_{\{c \leftarrow L_2\}}$.
   The \emph{iterated $c$-product} of a tree language $L$ over $\Sigma$ is the tree language $L^{n_c}$ recursively defined by:
   \begin{itemize}
     \item $L^{0_c}=\{c\}$,
     \item $L^{{(n+1)}_c}=L^{n_c}\cup L\cdot_{c} L^{n_c}$.
   \end{itemize}
   The $c$-\emph{closure} of the tree language $L$ is the language $L^{*_c}$ defined by $\bigcup_{n\geq 0} L^{n_c}$.
   
   In the following, we make use of weighted tree automata in order to compute tree kernels.
   See~\cite{tata} for details about classical tree automata. 
     
   Let $t$ be a tree over an alphabet $\Sigma$.   
   The tree $t^\sharp$ is obtained by indexing the symbols of $t$ by its position in a prefix course.
   We denote by $\Sigma_{t^\sharp}$ the set of the indexed symbols that appears in $t^\sharp$.
   The function $\mathrm{h}$ is the dual function, which drops the indexes ($\mathrm{h}(t^\sharp)=t$).
   Notice that the function $\mathrm{h}$ defines an equivalence relation over $T_{\Sigma_{t^\sharp}}$.
   Indeed, let $t_1$ and $t_2$ be two trees in $T_{\Sigma_{t^\sharp}}$.
   We define the relation $\sim_\mathrm{h}$ by $t_1\sim_\mathrm{h} t_2$ $\Leftrightarrow$ $\mathrm{h}(t_1)=\mathrm{h}(t_2)$.
   Since $\sim_\mathrm{h}$ is a relation based on the equality of images by $\mathrm{h}$, it can be shown that
  
   \begin{lemma}
     The relation $\sim_\mathrm{h}$ is an equivalence relation.
   \end{lemma}
   
  Let $\Sigma$ be an alphabet and $t=f(t_1,\ldots,t_k)$ be a tree in $T_{\Sigma}$.
  
  The set $\mathrm{SubTree}(t)$ is the set inductively defined by $\mathrm{SubTree}(t)= \{t\}\cup \bigcup_{1\leq j\leq k} \mathrm{SubTree}(t_j)$.
  Let $L$ be a tree language over $\Sigma$.
  The set $\mathrm{SubTreeSet}(L)$ is the set defined by $\mathrm{SubTreeSet}(L)=\bigcup_{t\in L}\mathrm{SubTree}(t)$.
    
  The formal tree series $\mathrm{SubTreeSeries}_t$ is the tree series over $\mathbb{N}$ inductively defined by $\mathrm{SubTreeSeries}_t=t+\sum_{1\leq j \leq k} \mathrm{SubTreeSeries}_{t_j}$.
  Let $L$ be a finite tree language.
  The formal tree series $\mathrm{SubTreeSeries}_L$ is the tree series over $\mathbb{N}$ defined by $\mathrm{SubTreeSeries}_L=\sum_{t'\in L} \mathrm{SubTreeSeries}_{t'}$. 
  Let us notice that if $L$ is not finite, since $\Sigma$ is a finite set of symbols, there exists a tree $t$ in $\Sigma_0$ that appears an infinite times as a subtree in $L$; thus $\mathrm{SubTreeSeries}_L$ is a tree series over $\mathbb{N}\cup\{+\infty\}$.
  
  \begin{definition}
    Let $L_1$ and $L_2$ be two finite tree languages.
    The \emph{subtree series kernel} of $L_1$ and $L_2$ is the integer $\mathrm{KerSeries}(L_1,L_2)$ defined by:
    
    \centerline{
      $\mathrm{KerSeries}(L_1,L_2)=\sum_{t\in T_\Sigma}(\mathrm{SubTreeSeries}_{L_1}\times \mathrm{SubTreeSeries}_{L_2})(t)$.
    }
  \end{definition}
  
  \begin{example}\label{ex subset tree}
    Let $\Sigma$ be the graded alphabet defined by $\Sigma_0=\{a,b\}$, $\Sigma_1=\{h\}$ and $\Sigma_2=\{f\}$.
    Let us consider the trees $t_1=f(h(a),f(h(a),b))$, $t_2=f(h(a),h(b))$ and $t_3=f(f(b,h(b)),f(h(a),h(b)))$.
    Then it can be shown that:
    \begin{itemize}
      \item $\mathrm{SubTree}(t_1)=\{t_1,f(h(a),b),h(a),a,b\}$
      \item $\mathrm{SubTree}(t_2)=\{t_2,h(a),h(b),a,b\}$
      \item $\mathrm{SubTreeSeries}_{t_1}=\mathbb{P}_{t_1}=t+f(h(a),b)+2h(a)+2a+b$
      \item $\mathrm{SubTreeSeries}_{t_2}=\mathbb{P}_{t_2}=t_2+h(b)+h(a)+a+b$
      \item $\mathrm{SubTreeSeries}_{t_3}=\mathbb{P}_{t_3}=t_3+f(b,h(b))+t_2+2h(b)+h(a)+3b+a$
      \item $\mathrm{SubTreeSeries}_{\{t_1,t_2\}}=\mathbb{P}_{\{t_1,t_2\}}=\mathbb{P}_{t_1}+\mathbb{P}_{t_2}=t+t_2+f(h(a),b)+3h(a)+h(b)+3a+2b$
      \item $\mathbb{P}_{\{t_1,t_2\}}\times \mathbb{P}_{\{t_3\}}=t_2+2h(b)+3h(a)+6b+3a$
      \item $\mathrm{KerSeries}(\{t_1,t_2\},\{t_3\})=15$
    \end{itemize}
  \end{example}

\section{Tree Series and Root-Weighted Tree Automata} \label{sec tree ser rwta} 
  
  A \emph{formal tree series}~\cite{BR82,EK03} $\mathbb{P}$ over a set $S$ is a mapping from $T_\Sigma$ to $S$.
  Let $\mathbb{M}=(M,+)$ be a monoid which identity is $0$.
  The \emph{support} of $\mathbb{P}$ is the set $\mathrm{Support}(\mathbb{P})=\{t\in T_\Sigma\mid \mathbb{P}(t)\neq 0\}$.
  Any formal tree series can be view as a formal sum $\mathbb{P}=\sum_{t\in T_\Sigma} (\mathbb{P}(t),t)$.
  In this case, the formal sum is considered associative and commutative. 
  
  Formal tree series can be realized by weighted tree automata.
  Weighted tree automata were defined over semirings~\cite{DPV05} or multioperator monoids~\cite{FMV09}.
  In this paper, we use particular automata, the weights of which belong to a monoid or a semiring, and only label the finality of states.
  Consequently, the automata we use are a strcit subclasses of weighted tree automata, with particular properties. 
  
  \subsection{Root-Weighted Tree Automata}

  \begin{definition}
    Let $\mathbb{M}=(M,+)$ be a commutative monoid.
    A $\mathbb{M}$-\emph{Root Weighted Tree Automata} ($\mathbb{M}$-RWTA) is a 4-tuple $(\Sigma,Q,\nu,\delta)$ where:
    \begin{itemize}
      \item $\Sigma=\bigcup_{k\in\mathbb{N}} \Sigma_k$ is a graded alphabet,
      \item $Q$ is a finite set of \emph{states},
      \item $\nu$ is a function from $Q$ to $M$ called the \emph{root weight function},
      \item $\delta$ is a subset of $Q\times\Sigma_k\times Q^k$, called the \emph{transition set}.
    \end{itemize}
  \end{definition}
  
  When there is no ambiguity, a $\mathbb{M}$-RWTA is called a RWTA.
  
  The root weight function $\nu$ is extended to $2^Q \rightarrow M$ for any subset $S$ of $Q$ by $\nu(S)=\sum_{s\in S} \nu(s)$.
  The function $\nu$ is equivalent to the finite subset of $Q\times M$ defined for any couple $(q,m)$ in $Q\times M$ by $(q,m)\in\nu$ $\Leftrightarrow$ $\nu(q)=m$.
  
  The transition set $\delta$ is equivalent to the function from $\Sigma_k\times Q^k$ to $2^Q$ defined for any symbol $f$ in $\Sigma_k$ and for any $k$-tuple $(q_1,\ldots,q_k)$ in $Q^k$ by $q\in\delta(f,q_1,\ldots,q_k) \Leftrightarrow (q,f,q_1,\ldots,q_k)\in\delta$. 
  The function $\delta$ is extended to $\Sigma_k \times (2^Q)^k  \rightarrow 2^Q$  as follows: for any symbol $f$ in $\Sigma_k$, for any $k$-tuple $(Q_1,\ldots,Q_k)$ of subsets of $Q$, $\delta(f,Q_1,\ldots,Q_k)=\bigcup_{(q_1,\ldots,q_k)\in Q_1\times\cdots\times Q_k} \delta(f,q_1,\ldots,q_k)$. 
  Finally, the function $\Delta$ is the function from  $T_{\Sigma}$ to $2^Q$ defined for any tree $t=f(t_1,\ldots,t_k)$ in $T_{\Sigma}$  by $\Delta(t)=\delta(f,\Delta(t_1),\ldots,\Delta(t_k))$.
  
  \noindent A \emph{weight} of a tree associated with $A$ is $\nu(\Delta(t))$.  
  The \emph{formal tree series realized} by $A$ is the formal tree series over $M$ denoted by $\mathbb{P}_A$ and defined by $\mathbb{P}_A(t)=\nu(\Delta(t))$, with $\nu(\emptyset)=0$ with $0$ the identity of $\mathbb{M}$.
  
  \begin{example}\label{ex rwta}
    Let us consider the graded alphabet $\Sigma$ defined by $\Sigma_0=\{a\}$, $\Sigma_1=\{h\}$ and $\Sigma_2=\{f\}$.
    Let $\mathbb{M}=(\mathbb{N},+)$.
    The RWTA $A=(\Sigma,Q,\nu,\delta)$ defined by
    \begin{itemize}
      \item $Q=\{1,2,3,4,5\}$,
      \item $\nu=\{(1,0),(2,3),(3,1),(4,2),(5,4)\}$,
      \item $\delta=\{(1,a),(3,a)(2,f,1,3),(4,f,3,3),(5,h,2),(5,h,4),(5,h,5)\}$,
    \end{itemize}
    is represented in Figure~\ref{fig ex RWTA} and realized the tree series:
    
    \centerline{
      $\mathbb{P}_A=a+5f(a,a)+4h(f(a,a))+4h(h(f(a,a)))+\cdots+4h(h(\ldots h(f(a,a))\ldots))+\cdots$.
    } 
  \end{example}  
  
  \begin{figure}[H]
    \centerline{
	    \begin{tikzpicture}[node distance=2.5cm,bend angle=30,transform shape,scale=1]
	      \node[state] (q5) {$5$};
	      \node[state, below left of=q5] (q2) {$2$};	
	      \node[state, below right of=q5] (q4) {$4$};	
	      \node[state, below of=q2] (q1) {$1$};	
	      \node[state, below of=q4] (q3) {$3$};	
        \draw (q5) ++(-0.75cm,0cm) edge[<-] node {$4$} (q5); 
        \draw (q2) ++(-0.75cm,0cm) edge[<-] node {$3$} (q2);  
        \draw (q4) ++(0.75cm,0cm) edge[above,<-] node {$2$} (q4);  
        \draw (q3) ++(0.75cm,0cm) edge[above,<-] node {$1$} (q3); 
        \draw (q1) ++(0cm,-1cm) node {$a$}  edge[->] (q1); 
        \draw (q3) ++(0cm,-1cm) node {$a$}  edge[->] (q3);  
        \path[->]
          (q2) edge[->,above left] node {$h$} (q5)
          (q4) edge[->,above right] node {$h$} (q5)
          (q5) edge[loop,->,above] node {$h$} ()
        ;
        \draw (q1) ++(1cm,1.5cm)  edge[->] node[above right,pos=0] {$f$} (q2) edge[shorten >=0pt,] (q1) edge[shorten >=0pt,] (q3); 
        \draw (q3) ++(0cm,1.5cm)  edge[->] node[above right,pos=0] {$f$} (q4) edge[shorten >=0pt,in=135,out=-135] (q3) edge[shorten >=0pt,in=45,out=-45] (q3); 
      \end{tikzpicture}
    }
    \caption{The RWTA $A$.}
    \label{fig ex RWTA}
  \end{figure}
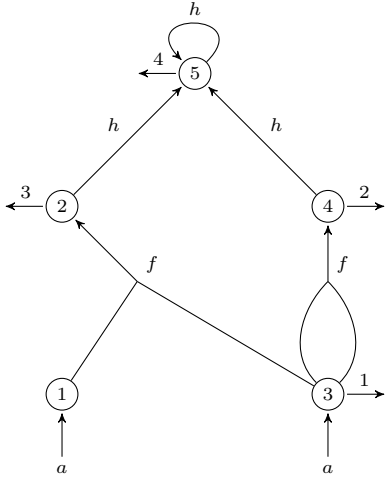
  
  Let $A_1=(\Sigma,Q_1,\nu_1,\delta_1)$ and $A_2=(\Gamma,Q_2,\nu_2,\delta_2)$ be two RWTAs. 
  A function $\mu$ is a \emph{morphism  of RWTA} from $A_1$ to $A_2$ if:
  \begin{itemize}
    \item $\forall q\in Q_1$, $\mu(q)\in Q_2$,
    \item $\forall f\in \Sigma_k$, $\mu(f)\in \Gamma_k$,
    \item $\forall (q,f,q_1,\ldots,q_k)\in\delta_1$, $(\mu(q),\mu(f),\mu(q_1),\ldots,\mu(q_k))\in\delta_2$,
    \item $\forall q\in Q_1$, $\nu_2(\mu(q))=\nu_1(q)$.
  \end{itemize}
  A morphism $\mu$ from $A_1$ to $A_2$ is said to be an \emph{isomorphism} if there exists a morphism $\mu^{-1}$ from $A_2$ to $A_1$.
  In this case, $A_1$ and $A_2$ are said to be \emph{isomorphic}.
  It can be shown by induction over the structure of any tree $t$ in $T_\Sigma$ that if $A_1$ and $A_2$ are isomorphic w.r.t. a morphism $\mu$ then $\Delta_2(\mu(t))=\bigcup_{q\in \Delta_1(t)}\{\mu(q)\}$. 
  Therefore 
  
  \begin{lemma}
    Let $A_1$ be a RWTA over an alphabet $\Sigma$.
    Let $A_2$ be a RWTA isomorphic to $A_1$ w.r.t. a morphism $\mu$.
    Then for any tree $t$ in $T_\Sigma$,
    
    \centerline{
      $\mathbb{P}_{A_1}(t)=\mathbb{P}_{A_2}(\mu(t))$.
    }
  \end{lemma}
  
  As a direct corollary, it holds
  
  \begin{corollary}  
    Two isomorphic RWTAs over the same alphabet realize the same tree series.  
  \end{corollary}  
  
  \subsection{RWTA Sequentialization}
  
  The RTWA $A$ is said to be \emph{sequential} if and only if for any tree $t$ in $T_\Sigma$, $\mathrm{Card}(\Delta(t))\leq 1$.
  Unlike the case of classical weighted tree and word automata, the RWTAs are sequentializable.
  
  \begin{theorem}\label{thm sequentialization ok}
    For any RTWA $A$, there exists a sequential RTWA $A'$ such that $\mathbb{P}_A=\mathbb{P}_{A'}$.
  \end{theorem}
  
  In order to prove Theorem~\ref{thm sequentialization ok}, let us define the subset construction~\cite{RS59} for any RWTA.
  
  \begin{definition}
    Let $A=(\Sigma,Q,\nu,\delta)$ be a RWTA. 
    The \emph{sequential RWTA associated with} $A$ is the RTWA $A'=(\Sigma,2^Q,\nu',\delta')$ defined by:
    \begin{itemize}
      \item $\forall S\subset Q$, $\nu'(S)=\sum_{s\in S} \nu(s)$;
      \item $\forall f\in\Sigma_k$, $\forall Q_1,\ldots,Q_k \subset Q$, $\delta'(f,Q_1,\ldots,Q_n)=\{\delta(f,Q_1,\ldots,Q_k)\}$.
    \end{itemize}
  \end{definition}
  
  Notice that $\nu'$ is equal to the extension of $\nu$ over the subsets of $Q$. 
  However, $\delta'$ is not equal to the extension of $\delta$ over sets since it necessarily returns a singleton.
  
  \begin{lemma}\label{lem delta' eq delta}
    Let $A=(\Sigma,Q,\nu,\delta)$.
    Let $A'=(\Sigma,2^Q,\nu',\delta')$ be the sequential RWTA associated with $A$. 
    For any tree $t$ in $T_\Sigma$,
    
    \centerline{
      $\Delta'(t)=\{\Delta(t)\}$. 
    }
  \end{lemma}
  \begin{proof}
    By definition of $\Delta'$, $\Delta'(f(t_1,\ldots,t_k))=\delta'(f,\Delta'(t_1),\ldots,\Delta'(t_k))$.
    \begin{enumerate}
      \item If $k=0$, then $\Delta'(f)=\delta'(f)$. 
      Moreover, by definition of $A'$, $\delta'(f)=\{\delta(f)\}$. 
      Since by definition of $\Delta$, $\delta(f)=\Delta(f)$, it holds that $\Delta'(f)=\{\Delta(f)\}$.
      \item Suppose that $k\neq 0$. 
      According to induction hypothesis, it holds that $\Delta'(f(t_1,\ldots,t_k))=\delta'(f,\{\Delta(t_1)\},\ldots,\{\Delta(t_k)\})$. 
      By definition of $\delta'$, $\delta'(f,\{\Delta(t_1)\},\ldots,\{\Delta(t_k)\})=\{\delta(f,\Delta(t_1),\ldots,\Delta(t_k))\}$, that equals by definition $\{\Delta(f(t_1,\ldots,t_k))\}$.
      Hence $\Delta'(t)=\{\Delta(t)\}$.
    \end{enumerate}
    \cqfd
  \end{proof}
  
  \begin{proposition}\label{prop seq pres ser}
    Let $A$ be a RWTA and $A'$ be the sequential RWTA associated with $A$. 
    Then:
    
    \centerline{
      $A'$ is a sequential RTWA that realizes $\mathbb{P}_A$.
    }
  \end{proposition}
  \begin{proof}
    Let $A=(\Sigma,Q,\nu,\delta)$ and $A'=(\Sigma,2^Q,\nu',\delta')$. 
    Let $t=f(t_1,\ldots,t_k)$ be a tree in $\Sigma$.
    According to Lemma~\ref{lem delta' eq delta}, $\Delta'(t)=\{\Delta(t)\}$. 
    As a direct consequence, $\mathrm{Card}(\Delta'(t))=1$ (since the state $\emptyset$ may be reached) and $\mathbb{P}_{A'}(t)=\nu'(\Delta'(t))=\nu'(\Delta(t))=\mathbb{P}_{A}(t)$.
    Hence $A'$ is a sequential RTWA that realizes $\mathbb{P}_A$.
    \cqfd
  \end{proof}
  
  \begin{example}
    Let us consider the RWTA defined in Example~\ref{ex rwta}.
    The sequential RWTA associated with $A$ is represented in Figure~\ref{fig ex seq rwta}.
  \end{example}  
  
  \begin{figure}[H]
    \centerline{
	    \begin{tikzpicture}[node distance=2.5cm,bend angle=30,transform shape,scale=1]
	      \node[state] (q5) {$5$};
	      \node[state, below of=q5, rounded rectangle] (q24) {$\{2,4\}$};	
	      \node[state, below of=q24, rounded rectangle] (q13) {$\{1,3\}$};	
        \draw (q5) ++(-0.75cm,0cm) edge[<-] node {$4$} (q5); 
        \draw (q24) ++(-1cm,0cm) edge[<-] node {$5$} (q24);  
        \draw (q13) ++(-1cm,0cm) edge[above,<-] node {$1$} (q13); 
        \draw (q13) ++(0cm,-1cm) node {$a$}  edge[->] (q13); 
        \path[->]
          (q24) edge[->,above left] node {$h$} (q5)
          (q5) edge[loop,->,above] node {$h$} ()
        ;
        \draw (q13) ++(0cm,1.5cm)  edge[->] node[above right,pos=0] {$f$} (q24) edge[shorten >=0pt,in=135,out=-135] (q13) edge[shorten >=0pt,in=45,out=-45] (q13); 
      \end{tikzpicture}
    }
    \caption{The sequential RWTA associated with $A$.}
    \label{fig ex seq rwta}
  \end{figure}
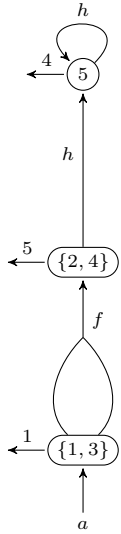
  
  Since a sequential RWTA is a RTWA, the set of tree series realized by a RTWA is closed under sequentialization, whatever the set of weights is.
  Let us now show that this set is also closed under several algebraic operations.
  
  \subsection{Sum and Product Closures}
  
    If $(M,+)$ is a commutative monoid, then the set of tree series over $M$ realized by a RTWA is closed under the sum.    
    
    \begin{definition}
      Let $A_1=(\Sigma,Q_1,\nu_1,\delta_1)$ and $A_2=(\Sigma,Q_2,\nu_2,\delta_2)$ be two RWTAs such that $Q_1\cap Q_2=\emptyset$. 
      The RWTA $A_1+A_2$ is the RWTA $A'=(\Sigma,Q_1\cup Q_2,\nu',\delta_1\cup\delta_2)$ where $\nu'$ is the function defined for any state $q$ in $Q_1\cup Q_2$ by:
       
      \centerline{
        $\nu'(q)=
          \left\{
            \begin{array}{l@{\ }l}
              \nu_1(q) & \text{ if }q\in Q_1,\\
	            \nu_2(q) & \text{ otherwise,}
	          \end{array}
	        \right.
	      $
	    }
    \end{definition}
    
    Notice that if $Q_1$ and $Q_2$ are not disjoint, then $Q_2$ can be changed using an isomorphism. 
    
    \begin{proposition}\label{prop sum aut is sum ser}
      Let $A_1$ and $A_2$ be two RWTAs. 
      Then for any tree $t$ in $T_\Sigma$:
      
      \centerline{
        $\mathbb{P}_{A_1+A_2}(t)=\mathbb{P}_{A_1}(t)+\mathbb{P}_{A_2}(t)$.
      }
    \end{proposition}
    \begin{proof}
      Let $A_1=(\Sigma,Q_1,\nu_1,\delta_1)$, $A_2=(\Sigma,Q_2,\nu_2,\delta_2)$ and $A'=A_1+A_2=(\Sigma,Q',\nu',\delta')$.
      Let $t=f(t_1,\ldots,t_k)$ be a tree in $T_\Sigma$.
      Let us first show by induction over the structure of $t$ that $\Delta'(t)=\Delta_1(t)\cup\Delta_2(t)$.
      By definition, $\Delta'(f(t_1,\ldots,t_k))=\delta'(f,\Delta'(t_1),\ldots,\Delta'(t_k))$.
      \begin{enumerate}
        \item if $k=0$, then $\Delta'(f)=\delta'(f)$. 
        By definition of $A'$, $\delta'(t)=\delta_1(t)\cup\delta_2(t)$ that equals by definition to $\Delta_1(f)\cup\Delta_2(f)$. 
        Hence $\Delta'(t)=\Delta_1(t)\cup\Delta_2(t)$.
        \item If $k\neq 0$, then by induction hypothesis, $\Delta'(f(t_1,\ldots,t_k))=\delta'(f,\Delta_1(t_1)\cup\Delta_2(t_1),\ldots,\Delta_1(t_k)\cup\Delta_2(t_k))$.
        Since $Q_1$ and $Q_2$ are disjoint, there is no transition $(q,f,q_1,\ldots,q_n)$ in $\Delta'$ such that there exists two integers $i$ and $j$ such that $q_i\in Q_1$ and $q_j\in Q_2$.
        Therefore $\delta'(f,\Delta_1(t_1)\cup\Delta_2(t_1),\ldots,\Delta_1(t_k)\cup\Delta_2(t_k))=\delta_1(f,\Delta_1(t_1),\ldots,\Delta_1(t_k))\cup \delta_2(f,\Delta_2(t_1),\ldots,\Delta_2(t_k))$, that is equal to $\Delta_1(f(t_1,\ldots,t_k))\cup \Delta_2(f(t_1,\ldots,t_k))$.  
        Hence $\Delta'(t)=\Delta_1(t)\cup\Delta_2(t)$.
      \end{enumerate}
      As a direct consequence, $\mathbb{P}_{A'}(t)=\nu'(\Delta_1(t)\cup\Delta_2(t))=\nu_1(\Delta_1(t))+\nu_2(\Delta_2(t))=\mathbb{P}_{A_1}(t)+\mathbb{P}_{A_2}(t)$.
      \cqfd
    \end{proof}
  
    A \emph{semiring} is a $5$-tuple $\mathbb{K}=(K,+,\times,0,1)$ such that:
    \begin{itemize}
      \item $(K,+)$ is a commutative monoid the identity of which is $0$,
      \item $(K,*)$ is a monoid the identity of which is $1$,
      \item $0\times \alpha=\alpha\times 0=0$ for any $\alpha$ in $K$,
      \item $\times$ distributes over $+$.
    \end{itemize}
    
    In the following, we consider trees over an alphabet $\Sigma$ and tree series over the semiring $\mathbb{K}$.    
    From this structure, another stable operation can be defined for formal tree series over $K$.
    
    Let $\mathbb{P}_1$ and $\mathbb{P}_2$ be two tree series.    
    The \emph{product of} $\mathbb{P}_1$ and $\mathbb{P}_2$ is the series $\mathbb{P}_1\times\mathbb{P}_2$ defined for any tree $t$ by $\mathbb{P}_1\times\mathbb{P}_2(t)=\mathbb{P}_1(t)\times \mathbb{P}_2(t)$.
    Let us show now that the product can be performed \emph{via} RWTAs.
    
    \begin{definition}
      Let $A_1=(\Sigma,Q_1,\nu_1,\delta_1)$ and $A_2=(\Sigma,Q_2,\nu_2,\delta_2)$ be two RWTAs. 
      The RWTA $A_1\times A_2$ is the RWTA $A'=(\Sigma,Q'=Q_1\times Q_2,\nu',\delta')$ defined by:
      \begin{itemize}
        \item $\forall f\in \Sigma_k$, $\forall q_1=(q_{1_1},q_{2_1}),\ldots,q_k=(q_{1_k},q_{2_k})\in Q'$, $\delta'(f,q_1,\ldots,q_k)=\delta_1(f,q_{1_1},\ldots,q_{1_k})\times \delta_2(f,q_{2_1},\ldots,q_{2_k})$,
        \item $\forall q=(q_1,q_2)\in Q'$, $\nu'(q)=\nu_1(q_1)\times\nu_2(q_2)$.
      \end{itemize}
    \end{definition}
    
    \begin{lemma}\label{lem delta pour prod cart}
      Let $A_1=(\Sigma,Q_1,\nu_1,\delta_1)$ and $A_2=(\Sigma,Q_2,\nu_2,\delta_2)$ be two RWTAs. 
      Then for any tree $t$ in $T_\Sigma$:
      
      \centerline{
        $\Delta'(t)=\Delta_1(t)\times \Delta_2(t)$.
      }
    \end{lemma}
    \begin{proof}
      By induction over the structure of $t=f(t_1,\ldots,t_k)$.
      By definition, $\Delta'(f(t_1,\ldots,t_k))=\delta'(f,\Delta'(t_1),\ldots,\Delta'(t_k))$.
      \begin{enumerate}
        \item if $k=0$, then $\Delta'(f)=\delta'(f)$. 
        By definition of $A'$, $\delta'(t)=\delta_1(t)\times\delta_2(t)$ that equals by definition to $\Delta_1(f)\times\Delta_2(f)$. 
        Hence $\Delta'(t)=\Delta_1(t)\times\Delta_2(t)$.
        \item If $k\neq 0$, then by induction hypothesis, $\Delta'(f(t_1,\ldots,t_k))=\delta'(f,\Delta_1(t_1)\times\Delta_2(t_1),\ldots,\Delta_1(t_k)\times\Delta_2(t_k))$.
        According to the definition of $\delta'$, 
        
        \centerline{        
          $\delta'(f,\Delta_1(t_1)\times\Delta_2(t_1),\ldots,\Delta_1(t_k)\times\Delta_2(t_k))=\bigcup_{q_j=(q_{1_j},q_{2_j})\in\Delta_1(t_j)\times\Delta_2(t_j),1\leq j\leq k} \delta'(f,q_1,\ldots,q_k)$.
        }
         
        By definition of $A'$, $ \delta'(f,q_1,\ldots,q_k)=\delta_1(f,q_{1_1},\ldots,q_{1_k})\times \delta_2(f,q_{2_1},\ldots,q_{2_k})$, for any $q_j=(q_{1_j},q_{2_j})\in\Delta_1(t_j)\times\Delta_2(t_j),1\leq j\leq k$.
        Furthermore, by definition of the cartesian product of set,
        
        \centerline{
          \begin{tabular}{l@{\ }l}
            $\Delta'(t)$ & $=\bigcup_{q_j=(q_{1_j},q_{2_j})\in\Delta_1(t_j)\times\Delta_2(t_j),1\leq j\leq k} \delta_1(f,q_{1_1},\ldots,q_{1_k})\times \delta_2(f,q_{2_1},\ldots,q_{2_k})$\\
            & $=\bigcup_{q_j\in\Delta_1(t_j),1\leq j\leq k} \delta_1(f,q_1,\ldots,q_k)\times \bigcup_{q_j\in\Delta_2(t_j),1\leq j\leq k) \delta_2(f,q_1,\ldots,q_k})$\\
          \end{tabular}
        }
                
        that is equal to $\delta_1(f,\Delta_1(t_1),\ldots,\Delta_1(t_k))\times \delta_2(f,\Delta_2(t_1),\ldots,\Delta_2(t_k))=\Delta_1(t)\times\Delta_2(t)$ by definition.
      \end{enumerate}
      \cqfd
    \end{proof}
    
    \begin{proposition}
      Let $A_1$ and $A_2$ be two RWTAs. 
      Then for any tree $t$ in $T_\Sigma$:
      
      \centerline{
        $\mathbb{P}_{A_1\times A_2}(t)=\mathbb{P}_{A_1}(t)\times \mathbb{P}_{A_2}(t)$.
      }
    \end{proposition}
    \begin{proof}
      Let $A_1=(\Sigma,Q_1,\nu_1,\delta_1)$, $A_2=(\Sigma,Q_2,\nu_2,\delta_2)$ and $A'=A_1\times A_2=(\Sigma,Q',\nu',\delta')$.
      Let $t=f(t_1,\ldots,t_k)$ be a tree in $T_\Sigma$.
      From Lemma~\ref{lem delta pour prod cart}, $\Delta'(t)=\Delta_1(t)\times \Delta_2(t)$.
      Hence $\mathbb{P}_{A'}(t)=\nu'(\Delta_1(t)\times\Delta_2(t))$. 
      By definition of $\nu'$, $\mathbb{P}_{A'}(t)=\sum_{q_1\in\Delta_1(t),q_2\in\Delta_2(t)} \nu_1(q_1)\times\nu_2(q_2)$.
      Since $\mathbb{K}$ is a semiring, by distibutivity, 
      $\mathbb{P}_{A'}(t)=(\sum_{q_1\in\Delta_1(t)} \nu_1(q_1))\times(\sum_{q_2\in\Delta_2(t)} \nu_2(q_2))$.
      Therefore, according to the definition of $ \Delta_1$ and $\Delta_2)$,
      $\mathbb{P}_{A'}(t)=\nu_1(\Delta_1(t))\times \nu_2(\Delta_2(t))$.
      \cqfd
    \end{proof}
  
  \begin{example}
    Let us consider the RWTA $A$ defined in Example~\ref{ex rwta} and let $A'$ be the RWTA represented in Figure~\ref{fig rwta autre}.
    The sum $A+A'$ is represented by the juxtaposition of Figure~\ref{fig ex RWTA} and Figure~\ref{fig rwta autre} and the product $A\times A'$ is represented in Figure~\ref{fig ex rwta prod}.
  \end{example} 
  
  \begin{minipage}{0.45\linewidth}
	  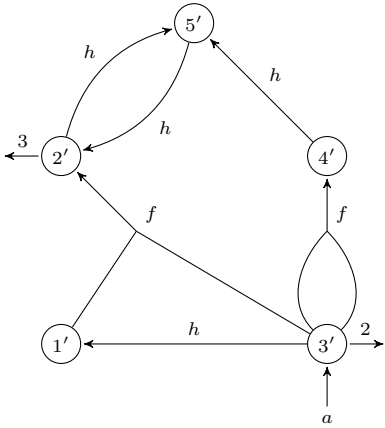
\begin{figure}[H]
	    \centerline{
		    \begin{tikzpicture}[node distance=2.5cm,bend angle=30,transform shape,scale=1]
		      \node[state] (q5) {$5'$};
		      \node[state, below left of=q5] (q2) {$2'$};	
		      \node[state, below right of=q5] (q4) {$4'$};	
		      \node[state, below of=q2] (q1) {$1'$};	
		      \node[state, below of=q4] (q3) {$3'$};	
	        \draw (q2) ++(-0.75cm,0cm) edge[<-] node {$3$} (q2);  
	        \draw (q3) ++(0.75cm,0cm) edge[above,<-] node {$2$} (q3); 
	        \draw (q3) ++(0cm,-1cm) node {$a$}  edge[->] (q3);  
	        \path[->]
	          (q2) edge[->,above left,bend left] node {$h$} (q5)
	          (q3) edge[->,above] node {$h$} (q1)
	          (q4) edge[->,above right] node {$h$} (q5)
	          (q5) edge[loop,->,below right,bend left] node {$h$} (q2)
	        ;
	        \draw (q1) ++(1cm,1.5cm)  edge[->] node[above right,pos=0] {$f$} (q2) edge[shorten >=0pt,] (q1) edge[shorten >=0pt,] (q3); 
	        \draw (q3) ++(0cm,1.5cm)  edge[->] node[above right,pos=0] {$f$} (q4) edge[shorten >=0pt,in=135,out=-135] (q3) edge[shorten >=0pt,in=45,out=-45] (q3); 
	      \end{tikzpicture}
	    }
	    \caption{The RWTA $A'$.}
	    \label{fig rwta autre}
	  \end{figure}
	\end{minipage}
	\hfill
	\begin{minipage}{0.45\linewidth}  
	  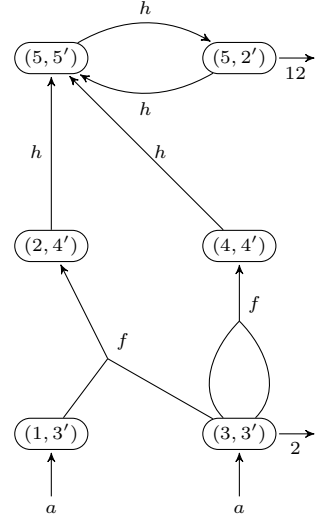
\begin{figure}[H]
	    \centerline{
		    \begin{tikzpicture}[node distance=2.5cm,bend angle=30,transform shape,scale=1]
		      \node[state,rounded rectangle] (q13') {$(1,3')$};
		      \node[state,rounded rectangle, right of=q13'] (q33') {$(3,3')$};	   
		      \node[state,rounded rectangle, above of=q13'] (q24') {$(2,4')$};	   
		      \node[state,rounded rectangle, above of=q33'] (q44') {$(4,4')$};	   
		      \node[state,rounded rectangle, above of=q24'] (q55') {$(5,5')$};	   
		      \node[state,rounded rectangle, right of=q55'] (q52') {$(5,2')$};	  
	        \draw (q33') ++(1cm,0cm) edge[<-] node {$2$} (q33'); 
	        \draw (q52') ++(1cm,0cm) edge[<-] node {$12$} (q52');  
	        \draw (q33') ++(0cm,-1cm) node {$a$}  edge[->] (q33');  
	        \draw (q13') ++(0cm,-1cm) node {$a$}  edge[->] (q13');  
	        \path[->]
	          (q24') edge[->,left] node {$h$} (q55')
	          (q44') edge[->,right] node {$h$} (q55')
	          (q55') edge[->,above,bend left] node {$h$} (q52')
	          (q52') edge[->,below,bend left] node {$h$} (q55')
	        ;
	        \draw (q13') ++(0.75cm,1cm)  edge[->] node[above right,pos=0] {$f$} (q24') edge[shorten >=0pt,] (q13') edge[shorten >=0pt,] (q33'); 
	        \draw (q33') ++(0cm,1.5cm)  edge[->] node[above right,pos=0] {$f$} (q44') edge[shorten >=0pt,in=135,out=-135] (q33') edge[shorten >=0pt,in=45,out=-45] (q33');
	      \end{tikzpicture}
	    }
	    \caption{The RWTA $A\times A'$.}
	    \label{fig ex rwta prod}
	  \end{figure}
	\end{minipage}
    
  Notice that series realized by RWTAs are not necessarily closed under classical regular operations.

  \subsection{Case of the $a$-Product}
    
    Let $a$ be a symbol in $\Sigma_0$.
    The $a$-\emph{product of} $\mathbb{P}_1$ and $\mathbb{P}_2$ is the series $\mathbb{P}_1\cdot_a\mathbb{P}_2$ defined for any tree $t$ by $\mathbb{P}_1\cdot_a \mathbb{P}_2(t)=\sum_{t_1,t_2\in T_\Sigma,t=t_1\cdot_a t_2} \nu_1(t_1)\times \nu_2(t_2)$. 
    
    Let us show that the $a$-product of two series realized by some RTWAs may not be realized by any RTWA.
    
    The \emph{image} of a tree series $\mathbb{P}$ is the set $\mathrm{Im}(\mathbb{P})=\{\alpha\in K\mid \exists t\in T_\Sigma, \mathbb{P}(t)=\alpha\}$.
    
    \begin{lemma}\label{lem im rwta finite}
      Let $A$ be a RWTA. Then:
      
      \centerline{
        $\mathrm{Im}(\mathbb{P}_A)$ is a finite set.
      }
    \end{lemma}
    \begin{proof}
      Let $A=(\Sigma,Q,\nu,\delta)$.
      By definition, for any tree $t$ in $T_\Sigma$, $\mathbb{P}_{A}(t)=\sum_{q\in\Delta(t)}\nu(q)$.
      Consequently, $\mathbb{P}_A(t)$ belongs to the subset $\{\alpha\in K\mid \exists S\subset Q, \alpha=\nu(S)\}$ of $K$.
      Therefore, $\mathrm{Card}(\mathrm{Im}(\mathbb{P}_A))$ is less than $2^{\mathrm{Card}(Q)}$.
    \end{proof}
    
    \begin{proposition}
      Let $\Sigma$ be an alphabet and $a$ be a symbol in $\Sigma_0$.
      There exist formal tree series $\mathbb{P}_1$ and $\mathbb{P}_2$ such that $\mathrm{Im}(\mathbb{P}_1\cdot_a\mathbb{P}_2)$ is not finite.
    \end{proposition}
    \begin{proof}
      Let $\mathbb{K}=(\mathbb{N},+,\times,0,1)$.
      Let us consider the alphabet $\Sigma$ defined by $\Sigma_0=\{a,b\}$, $\Sigma_2=\{f\}$.
      Let us consider the tree language $L$ defined by $(f(a,b))^{*_a}$.
      Let us consider the series $\mathbb{P}_1$ and $\mathbb{P}_2$ defined for any tree $t$ in $T_\Sigma$ as follows:
      \begin{itemize}
        \item 
          $\mathbb{P}_1(t)=
            \left\{
              \begin{array}{l@{\ }l}
                1 & \text{ if } t\in L,\\
                0 & \text{ otherwise;}\\
              \end{array}
            \right.
          $
        \item 
          $\mathbb{P}_2(t)=
            \left\{
              \begin{array}{l@{\ }l}
                1 & \text{ if } t=a,\\
                0 & \text{ otherwise;}\\
              \end{array}
            \right.
          $
      \end{itemize}
      Let $A_1=(\Sigma,Q_1,\nu_1,\delta_1)$ be the RTWA defined by:
      \begin{itemize}
        \item $Q_1=\{0,1\}$,
        \item $\nu_1(0)=1$, $\nu_1(1)=0$,
        \item $\delta_1(a)=\{0\}$, $\delta_1(b)=\{1\}$, $\delta_1(f,0,1)=\{f\}$.
      \end{itemize}
      Let $A_2=(\Sigma,Q_2,\nu_2,\delta_2)$ be the RTWA defined by:
      \begin{itemize}
        \item $Q_2=\{0\}$,
        \item $\nu_2(0)=2$,
        \item $\delta_2(a)=\{0\}$.
      \end{itemize}
      It can be checked that:
      \begin{enumerate}
        \item the series $\mathbb{P}_1$ is realized by the RTWA $A_1$;
        \item the series $\mathbb{P}_2$ is realized by the RTWA $A_2$;
        \item the series $\mathbb{P}_1\cdot_a \mathbb{P}_2$ associates any tree $t$ in $L$ with the integer $2^{\mathrm{h}(t)}$, where $\mathrm{h}(t)$ is the height of $t$.
      \end{enumerate}
      Since $L$ is infinite, so is $\mathrm{Im}(\mathbb{P}_A)$.
      According to Lemma~\ref{lem im rwta finite}, $\mathbb{P}_1\cdot_a \mathbb{P}_2$ can not be realized by any RWTA. 
      \cqfd
    \end{proof}
    
    \begin{corollary}
      Let $\Sigma$ an alphabet and $a$ be a symbol in $\Sigma_0$.
      The tree series realized by some RWTAs are not closed under $a$-product. 
    \end{corollary}
    
    The same reasoning can be applied on the case of iterated product.
    
  \subsection{Quotient of a RWTA}
  
  Morphisms of RWTAs can be applied w.r.t. an equivalence relation in order to define quotients of RWTA.
  
  Given an equivalence relation $\sim$ over a set $Q$, we denote by $Q_\sim$ the set of equivalence classes of $\sim$. 
  Given a state $q$ in $Q$, we denote by $[q]_\sim$ the equivalence class of $q$ w.r.t. $\sim$, \emph{i.e.} $\{q'\in Q\mid q'\sim q\}$.
  
  \begin{definition}
    Let $A=(\Sigma,Q,\nu,\delta)$ be a RWTA and $\sim$ be an equivalence relation over $\nu$. 
    The \emph{quotient of} $A$ w.r.t. $\sim$ is the RWTA $A_\sim=(\Sigma,Q_\sim,\nu',\delta')$ defined by:
    \begin{itemize}
      \item $\forall C \in Q_\sim$, $\nu'(C)=\sum_{q\in C} \nu(C)$,
      \item $\forall C_1,\ldots,C_{k+1}\in Q_\sim$, $C_{k+1} \in\delta'(f,C_1,\ldots,C_k)$ $\Leftrightarrow$ $\forall i\leq k+1$, $\exists q_i\in C_i$, $q_{k+1}\in\delta(f,q_1,\ldots,q_k)$
    \end{itemize}
  \end{definition}
  
  Notice that the quotient of a RWTA $A$ does not necessary realize the same series as $A$.
  Nevertheless, in the following of this paper, we use particular relation that preserves the series while quotienting.
  
  \begin{definition}
    Let $A=(\Sigma,Q,\nu,\delta)$ be a RWTA and $q$ be a state in $Q$. 
    The \emph{down language of} $q$ is the language $L_q(A)$ defined by:
    
    \centerline{
      $L_q(A)=\{t\in T_\Sigma\mid q\in\Delta(t)\}$.
    }
  \end{definition}
  
  \begin{proposition}\label{prop series real somme lang down}
    The tree series realized by a RWTA $A=(\Sigma,Q,\nu,\delta)$ is equal to $\sum_{q\in Q} \nu(q) L_q(A)$.
  \end{proposition}
  \begin{proof}
    By definition, it holds that $\mathbb{P}_A=\sum_{t\in T_\Sigma} \nu(\Delta(t)) t$.    
    Consequently, by definition of $\nu(\Delta(t))$, $\mathbb{P}_A=\sum_{t\in T_\Sigma}\sum_{q\in \Delta(t)}\nu(q) t$.
    Furthermore, since any tree $t$ such that $\Delta(t)$ is not empty is a tree that belongs to $L_q(t)$ for some state $q$ in $Q$, $\mathbb{P}_A=\sum_{q\in Q}\sum_{t\mid q\in\Delta(t)} \nu(q) t$.
    Thus, by definition of $L_q(A)$, $\mathbb{P}_A=\sum_{q\in Q}\sum_{t\in L_q(A)} \nu(q) t$.
    Consequently, since the coefficient $\nu(q)$ belongs to a semiring, by distributivity, $\mathbb{P}_A=\sum_{q\in Q} \nu(q) L_q(A)$.
    \cqfd
  \end{proof}
  
  \begin{definition}
    Let $A=(\Sigma,Q,\nu,\delta)$ be a RWTA. 
    Let $\sim$ be an equivalence relation over $Q$. 
    The relation $\sim$ is said to be \emph{down compatible} with $A$ if for any two states $q_1$ and $q_2$ in $Q$, it holds:
    
    \centerline{
      $q_1\sim q_2$ $\Rightarrow$ $L_{q_1}(A)=L_{q_2}(A)$.
    }  
  \end{definition}
  
  \begin{proposition}
    Let $A$ be a RWTA and $\sim$ be an equivalence relation down compatible with $A$. 
    Then:
    
    \centerline{
      $\mathbb{P}_A=\mathbb{P}_{A_\sim}$.
    }
  \end{proposition}
  \begin{proof}
    Let $A=(\Sigma,Q,\nu,\delta)$ and $A_\sim=(\Sigma,Q_\sim,\nu',\delta')$.
    According to Proposition~\ref{prop series real somme lang down}, $\mathbb{P}_A=\sum_{q\in Q} \nu(q) L_q(A)$ and ${\mathbb{P}_A}_\sim=\sum_{C\in Q_\sim} \nu'(C) L_C(A)$.
    By definition of $\nu'$, ${\mathbb{P}_A}_\sim=\sum_{C\in Q_\sim} (\sum_{q\in C} \nu(q)) L_C(A)$.
    Since $\sim$ is down compatible, for any state $C$ in $Q_\sim$, for any two states $q_1$ and $q_2$ in $C$, $L_{q_1}(A)=L_{Q_2}(A)$.
    Therefore ${\mathbb{P}_A}_\sim=\sum_{C\in Q_\sim} (\sum_{q\in C} \nu(q) L_q(A))$.
    Moreover, since $\sim$ is an equivalence relation, any state of $Q$ belongs to one and only one state $C$ in $Q_\sim$.
    Consequently, ${\mathbb{P}_A}_\sim=\sum_{q\in C} \nu(q) L_q(A)=\mathbb{P}_A$.
  \end{proof}  
  
  \begin{example}
    Let us consider the RWTA $A''=A\times A'$ represented in Figure~\ref{fig ex rwta prod}.
    Let us consider the equivalence relation $\sim$ over the set of states of $A''$ defined by $(q_1,q'_1)\sim(q_2,q'_2)$ $\Leftrightarrow$ $q'_1=q'_2$.
    It can be shown that $\sim$ is down compatible with $A$. 
    The quotient $A''_\sim$ is represented in Figure~\ref{fig ex quot}.
  \end{example} 
  
  \begin{figure}[H]
    \centerline{
	    \begin{tikzpicture}[node distance=2.5cm,bend angle=30,transform shape,scale=1]
	      \node[state,rounded rectangle] (q33') {$\{(1,3'),(3,3')\}$};	  
	      \node[state,rounded rectangle, above of=q13'] (q44') {$\{(2,4'),(4,4')\}$};	   
	      \node[state,rounded rectangle, above left of=q24'] (q55') {$\{(5,5')\}$};	   
	      \node[state,rounded rectangle, above right of=q24'] (q52') {$\{(5,2')\}$};	  
        \draw (q33') ++(1.5cm,0cm) edge[<-] node {$2$} (q33'); 
        \draw (q52') ++(1.25cm,0cm) edge[<-] node {$12$} (q52');  
        \draw (q33') ++(0cm,-1cm) node {$a$}  edge[->] (q33');  
        \path[->]
          (q44') edge[->,below left] node {$h$} (q55')
          (q55') edge[->,above,bend left] node {$h$} (q52')
          (q52') edge[->,below,bend left] node {$h$} (q55')
        ;
        \draw (q33') ++(0cm,1.5cm)  edge[->] node[above right,pos=0] {$f$} (q44') edge[shorten >=0pt,in=135,out=-135] (q33') edge[shorten >=0pt,in=45,out=-45] (q33');
      \end{tikzpicture}
    }
    \caption{The RWTA $A''_\sim$.}
    \label{fig ex quot}
  \end{figure}
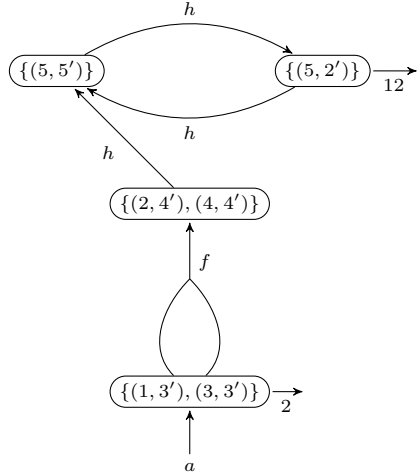
  
  Now that we have defined the notion of RWTA, let us apply it on tree kernel computations.
  
\section{Subtree Kernel}\label{sec subtree Ker}

  In this section, we show how to efficiently compute the subtree kernel of two finite tree languages using RWTAs.
  We first associate any tree with a RWTA that realizes its subtree series.
  
  \subsection{Subtree Automaton of a Tree}  
  
  \begin{definition}
    Let $\Sigma$ be an alphabet.
    Let $t$ be a tree in $T_\Sigma$.
    The \emph{subtree automaton} associated with $t$ is the RWTA $A_t=(\Sigma,Q,\nu,\delta)$ defined by:
    \begin{itemize}
      \item $Q=\mathrm{SubTreeSet}(t^\sharp)$,
      \item $\forall q\in Q$, $\nu(q)=1$,
      \item $\forall f\in \Sigma_{t^\sharp}$, $\forall t_1,\ldots,t_{k+1}\in Q$, $t_{k+1}\in\delta(\mathrm{h}(f),t_1,\ldots,t_k)$ $\Leftrightarrow$ $t_{k+1}=f(t_1,\ldots,t_k)$.
    \end{itemize}
  \end{definition}
  
  \begin{example}
    Let us consider the tree $t_1=f(h(a),f(h(a),b))$ defined in Example~\ref{ex subset tree}.
    Then $t^\sharp=f_1(h_2(a_3),f_4(h_5(a_6),b_7))$.
    The RWTA $A_{t_1}$ is represented in Figure~\ref{fig ex At}, where all the root weights, equal to $1$, are not represented.
  \end{example} 
  
  \begin{figure}[H]
    \centerline{
	    \begin{tikzpicture}[node distance=2.5cm,bend angle=30,transform shape,scale=1]
	      \node[state,rounded rectangle] (a3) {$a_3$};
	      \node[state, above of=a3,rounded rectangle] (h2) {$h_2(a_3)$};	
	      \node[state, below right of=a3,rounded rectangle] (a6) {$a_6$};	
	      \node[state, above of=a6,rounded rectangle] (h5) {$h_5(a_6)$};	
	      \node[state, above right of=h5,rounded rectangle] (f4) {$f_4(h_5(a_6),b_7)$};	
	      \node[state, right of=a6,rounded rectangle] (b7) {$b_7$};	   
	      \node[state, above left of=f4,rounded rectangle] (t1) {$t_1^\sharp$};	 
        \draw (a3) ++(0cm,-1cm) node {$a$}  edge[->] (a3);  
        \draw (a6) ++(0cm,-1cm) node {$a$}  edge[->] (a6);  
        \draw (b7) ++(0cm,-1cm) node {$b$}  edge[->] (b7);  
        \path[->]
          (a3) edge[->,below left] node {$h$} (h2)
          (a6) edge[->,below left] node {$h$} (h5)
        ;
        \draw (h5) ++(1cm,0.5cm)  edge[->] node[right,pos=0] {$f$} (f4) edge[shorten >=0pt] (h5) edge[shorten >=0pt] (b7);
        \draw (t1) ++(0cm,-1cm)  edge[->] node[above left,pos=0] {$f$} (t1) edge[shorten >=0pt] (h2) edge[shorten >=0pt] (f4);
      \end{tikzpicture}
    }
    \caption{The RWTA $A_{t_1}$.}
    \label{fig ex At}
  \end{figure}
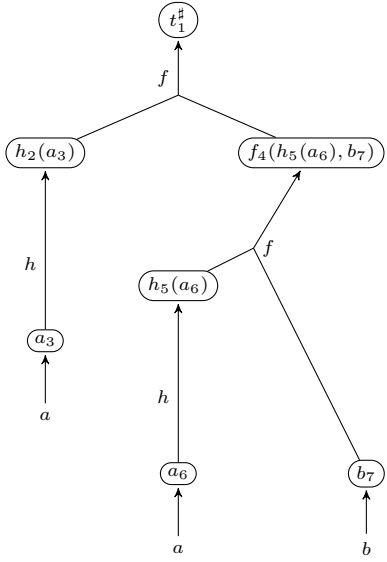
  
  \begin{lemma}\label{lem At bon serie}
    Let $\Sigma$ be an alphabet.
    Let $t$ be a tree in $T_\Sigma$
    Then:
    
    \centerline{
      $\mathbb{P}_{A_t}= \mathrm{SubTreeSeries}_t$.
    }
  
  \end{lemma}
  \begin{proof} 
    Let us set $A_t=(\Sigma,Q,\nu,\delta)$ and $A_{t_i}=(\Sigma,Q_i,\nu_i,\delta_i)$ for $1\leq i\leq k$. 
    Notice that by definition: 
    
    \centerline{
      $Q=\{t\}\cup \bigcup_{1\leq i\leq k}Q_i$ and $\delta=\{(t,f,t_1,\ldots,t_k)\}\cup \bigcup_{1\leq i\leq k}\delta_i$.
    }
    
    \noindent Consequently, $\mathbb{P}_{A_{t}}=t + \sum_{1\leq i \leq k} \mathbb{P}_{A_{t_i}}$.
    By definition, $\mathrm{SubTreeSeries}_t=t+\sum_{1\leq j \leq k} \mathrm{SubTreeSeries}_{t_j}$.  
    Furthermore, by induction hypothesis, $\mathbb{P}_{A_{t_i}}= \mathrm{SubTreeSeries}_{t_i}$.
    Therefore it holds that 
    
    \centerline{ $\mathbb{P}_{A_{t}}=t+\sum_{1\leq j \leq k} \mathrm{SubTreeSeries}_{t_j}=\mathrm{SubTreeSeries}_t$.}
    \cqfd
  \end{proof}
  
  Another RWTA can be defined in order to realize the subtree series associated with a tree.
  This RWTA needs less space since its states are exactly its subsets.
  
  \begin{definition}
    Let $\Sigma$ be an alphabet.
    Let $t$ be a tree in $T_\Sigma$.
    The \emph{sequential subtree automaton} associated with $t$ is the RWTA $\mathrm{seq}(A_t)=(\Sigma,Q,\nu,\delta)$ defined by:
    \begin{itemize}
      \item $Q=\mathrm{SubTreeSet}(t)$,
      \item $\forall t'\in Q$, $\nu(t')=\mathrm{SubTreeSeries}_t(t')$,
      \item $\forall f\in \Sigma$, $\forall t_1,\ldots,t_{k+1}\in Q$, $t_{k+1}\in\delta(f,t_1,\ldots,t_k)$ $\Leftrightarrow$ $t_{k+1}=f(t_1,\ldots,t_k)$.
    \end{itemize}
  \end{definition}
  
  \begin{example}
    Let us consider the tree $t_1=f(h(a),f(h(a),b))$ defined in Example~\ref{ex subset tree}.
    The RWTA $\mathrm{seq}(A_{t_1})$ is represented in Figure~\ref{fig ex seq At}.
  \end{example} 
  
  \begin{figure}[H]
    \centerline{
	    \begin{tikzpicture}[node distance=2.5cm,bend angle=30,transform shape,scale=1]
	      \node[state,rounded rectangle] (a) {$a$};
	      \node[state, above of=a,rounded rectangle] (h) {$h(a)$};		
	      \node[state, right of=a,rounded rectangle] (b) {$b$};	   
	      \node[state, right of=h,rounded rectangle] (f) {$f(h(a),b)$};	
	      \node[state, above left of=f,rounded rectangle] (t1) {$t_1$};	 
        \draw (a) ++(-0.75cm,0cm) edge[<-] node {$2$} (a); 
        \draw (h) ++(-1cm,0cm) edge[<-] node {$2$} (h); 
        \draw (b) ++(0.75cm,0cm) edge[<-] node {$1$} (b); 
        \draw (t1) ++(0cm,0.5cm) edge[<-] node {$1$} (t1);
        \draw (f) ++(1.5cm,0cm) edge[<-] node {$1$} (f);   
        \draw (a) ++(0cm,-1cm) node {$a$}  edge[->] (a);  
        \draw (b) ++(0cm,-1cm) node {$b$}  edge[->] (b);  
        \path[->]
          (a) edge[->,below left] node {$h$} (h)
        ;
        \draw (h) ++(1cm,0cm)  edge[->] node[above,pos=0] {$f$} (f) edge[shorten >=0pt] (h) edge[shorten >=0pt] (b);
        \draw (t1) ++(0cm,-1cm)  edge[->] node[above left,pos=0] {$f$} (t1) edge[shorten >=0pt] (h) edge[shorten >=0pt] (f);
      \end{tikzpicture}
    }
    \caption{The RWTA $\mathrm{seq}(A_{t_1})$.}
    \label{fig ex seq At}
  \end{figure}
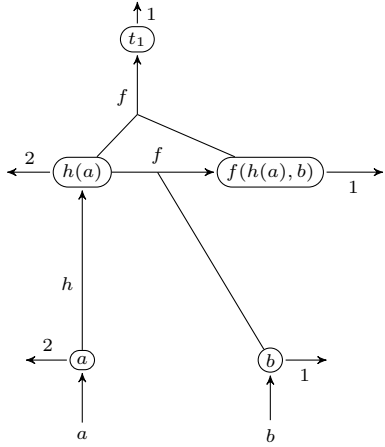
  
  However, the sequential subtree automaton needs the tree series to be known in order to compute it.
  Nevertheless, we show how to compute it from a quotient of the subtree automaton.
  Once this tree computed, it can be reduced using the equivalence $\sim_\mathrm{h}$. 
  Furthermore, this RWTA is isomorphic to the one obtained by subset construction.
  Consequently, we compute a sequential RWTA the number of states oh which is equal to the number of its different subtrees.
  We first show that $\sim_\mathrm{h}$ is down compatible with the subtree automaton, then we show that its application leads to the computation of the sequential subtree automaton.
  
  \begin{lemma}\label{lem arbre ht arriv dans t}
    Let $\Sigma$ be a graded alphabet.
    Let $t$ be a tree in $T_\Sigma$.
    Let $A_t=(\Sigma,Q,\nu,\delta)$.
    For any tree $r$ in $\mathrm{SubTreeSet}(t)$, it holds:
    
    \centerline{
      $\Delta(r)=\{r'\in Q\mid \mathrm{h}(r')=r\}$.
    }
  \end{lemma}
  \begin{proof}
    By induction over the structure of $r=f(r_1,\ldots,r_k)$.
    By definition of $\Delta$, $\Delta(r)=\delta(f,\Delta(r_1),\ldots,\Delta(r_k))$.
    By induction hypothesis, $\Delta(r_i)=\{r'_i\in Q\mid \mathrm{r'_i}=r_i\}$.
    Thus, $\Delta(r)=\{f_j(r'_1,\ldots,r'_k)\in Q\mid \mathrm{h}(f_j)=f\wedge r_i=\mathrm{h}(r'_i)\}$.
    Hence $\Delta(r)=\{r'\in Q\mid \mathrm{h}(r')=r\}$.
    \cqfd
  \end{proof}
  
  As a direct consequence, for any state $r$ in $Q$, any tree $r'$ in $L_r(A_t)$ satisfies $\mathrm{h}(r)=r'$.
  
  \begin{corollary}\label{cor arbre ht arriv dans t}
    Let $\Sigma$ be a graded alphabet.
    Let $t$ be a tree in $T_\Sigma$.
    Let $A_t=(\Sigma,Q,\nu,\delta)$.
    Then for any state $r$ in $Q$, $L_r(A_t)=\{\mathrm{h}(r)\}$.
  \end{corollary}
  
  \begin{lemma}
    Let $\Sigma$ be a graded alphabet.
    Let $t$ be a tree in $T_\Sigma$.
    Then:
    
    \centerline{
      $\sim_\mathrm{h}$ is down compatible with $A_t$.
    }
  \end{lemma}
  \begin{proof}
    Let $A_t=(\Sigma,Q,\nu,\delta)$.
    According to Corollary~\ref{cor arbre ht arriv dans t}, $L_r(A_t)=\{\mathrm{h}(r)\}$.
    Consequently, for any two states $r_1$ and $r_2$ in $Q$, $r_1\sim_\mathrm{h}r_2$ $\Rightarrow$ $\mathrm{h}(r_1)=\mathrm{h}(r_2)$ $\Rightarrow$ $L_{r_1}(A_t)=L_{r_2}(A_t)$.
    \cqfd
  \end{proof}
  
  \begin{proposition}
    Let $\Sigma$ be a graded alphabet.
    Let $t$ be a tree in $T_\Sigma$.
    Then:
    
    \centerline{
      The RWTA $\mathrm{seq}(A_t)$ is isomorphic to ${{A_t}}_{\sim_{h}}$.
    }    
  \end{proposition}
  \begin{proof}
    Let us set $A_t=(\Sigma,Q=\mathrm{SubTreeSet}(t^\sharp),\nu,\delta)$, ${{A_t}}_{\sim_{h}}=(\Sigma,Q_\sim,\nu',\delta')$  and $\mathrm{seq}(A_t)=(\Sigma,\mathrm{SubTreeSet}(t),$ $\nu'',\delta'')$.
    
    By definition, any state $C=\{t_1,\ldots\}$ in $Q_{\sim_{h}}$ can be associated with $\mathrm{h}(t_1)$ since any two states $q$ and $q'$ in $C$ satisfies by definition  $\mathrm{h}(q)=\mathrm{h}(q')$. 
    Consequently, let us consider the function $g$ that associates to any state $C=\{t_1,\ldots,\}$ in $Q_{\sim_{h}}$ the tree $\mathrm{h}(t_1)$.
    Notice that this function is bijective since for any tree $r$, $g^{-1}(r)=\{r'\in\mathrm{\mathrm{SubTreeSet}(t^\sharp)\mid \mathrm{h}(r')=r}\}$.
    Let us show that this function defines an isomorphism between ${{A_t}}_{\sim_{h}}$ and $\mathrm{seq}(A_t)$.
    \begin{enumerate}
      \item By definition, for any state $C$ in $Q_{\sim_{h}}$, $g(f)$ belongs to $\mathrm{SubTreeSet}(t)$.
      \item For any transition $(C_{k+1},f,C_1,\ldots,C_k)$ in $\delta'$, there exist by definition $t_1,\ldots,t_{k}$ in $Q$ and a symbol $f_j$ satisfying $\mathrm{h}(f_j)=f$ such that $(f_j(t_1,\ldots,t_k),f,t_1,\ldots,t_k)$ is in $\delta$. 
        Consequently $f_j(t_1,\ldots,t_k)$ is a subtree of $t^\sharp$ and then $f(\mathrm{h}(t_1),\ldots,\mathrm{h}(t_k))$ is a subtree of $t$. 
        Therefore by definition of $A$ $(g(C_{k+1})=\mathrm{h}(t_{k+1}),f,g(C_1)=\mathrm{h}(t_1),\ldots,g(C_k)=\mathrm{h}(t_k))$ is in $\delta''$.
      \item According to Corollary~\ref{cor arbre ht arriv dans t}, for any state $r$ in $Q$, $L_r(A_t)=\{\mathrm{h}(r)\}$.
        Consequently, for any state $C$ in $Q_\sim$, $\nu'(C)=\sum_{r \in C}\nu(r)$.
        According to Lemma~\ref{lem arbre ht arriv dans t}, for any tree $r'$, $\Delta(r')=\{r\mid \mathrm{r}=r'\}=C$. 
        Moreover, according to Lemma~\ref{lem At bon serie}, for any tree $r'$, $\nu(\Delta(r'))=\mathrm{SubTreeSeries}_t(r')$.
        Then $\nu'(C)=\mathrm{SubTreeSeries}_t(g(C))=\nu''(g(C))$.
    \end{enumerate}
    \cqfd
  \end{proof}  
  
  We denote by $|t|$ the size of a tree $t$, \emph{i.e.} the number of its nodes.
  Since the subtree automaton and the relation $\sim_\mathrm{h}$ can be computed in linear time, it holds that
  
  \begin{corollary}
    Let $\Sigma$ be an alphabet.
    Let $t$ be a tree in $T_\Sigma$.
    Then:
    
    \centerline{
      The RWTA $\mathrm{seq}(A_t)$ can be computed in time and space $O(|t|)$.
    }        
  \end{corollary}
  
  Let us now show that the sequential subtree automaton is a sequential RWTA.
  Let us prove it by showing that the computation of the accessible part of the subset construction leads exactly to the computation of the quotient of the subtree automaton, where the \emph{accessible part} of the sequential RWTA associated with a RWTA is the RWTA based on the states the down languages of which is not empty. 
  
  \begin{proposition}
    Let $\Sigma$ be an alphabet.
    Let $t$ be a tree in $T_\Sigma$.
    Then:
    
    \centerline{
      The accessible part of the sequential RWTA associated with $A_t$ is equal to ${{A_t}}_{\sim_{h}}$.
    }    
  \end{proposition}
  \begin{proof}
    Let us set $A_t=(\Sigma,Q=\mathrm{SubTreeSet}(t^\sharp),\nu,\delta)$, $A'=(\Sigma,Q',\nu',\delta')$ the accessible part of sequential RWTA associated with $A_t$ and $A''={{A_t}}_{\sim_{h}}=(\Sigma,Q'',\nu'',\delta'')$.
        
    According to Lemma~\ref{lem delta' eq delta}, $\Delta'(r')=\{\Delta(r')\}$.    
    According to Lemma~\ref{lem arbre ht arriv dans t}, $\Delta(r')=\{r\in \mathrm{SubTreeSet}(t^\sharp)\mid \mathrm{h}(r)=r'\}$.
    Hence, $\Delta'(r')=\{\{r\in \mathrm{SubTreeSet}(t^\sharp)\mid \mathrm{h}(r)=r'\}\}\}$ that is an equivalence class of $\sim_\mathrm{h}$.
    Therefore, $Q''=Q'$ and $\delta'=\delta''$.    
    Moreover, for any state $C$ in $Q'$, $\nu'(C)$ and $\nu''(C)$ are both equal by definition to $\sum_{c\in C}\nu(c)$.
    Consequently $A'=A''$.    
    \cqfd
  \end{proof}
  
  \begin{corollary}
    Let $\Sigma$ be an alphabet.
    Let $t$ be a tree in $T_\Sigma$.
    Then:
    
    \centerline{
      The accessible part of the sequential RWTA associated with $A_t$ is isomorphic to $\mathrm{seq}(A_t)$.
    }    
  \end{corollary}
  
  To sum up the properties of the sequential subtree RWTA associated with a tree:
  
  \begin{corollary}
    Let $\Sigma$ be an alphabet.
    Let $t$ be a tree in $T_\Sigma$.
    Then the RWTA $\mathrm{seq}(A_t)$:    
    \begin{itemize}
      \item is a sequential RWTA,
      \item is smaller than $A_t$,
      \item realizes $\mathrm{SubTreeSeries}_t$,
      \item is constructed in time and space $O(|t|)$.
    \end{itemize}
  \end{corollary}  
  
  Let us now show how to extend this construction to finite tree languages.
  
  \subsection{Subtree Automaton of a Finite Tree Language}  
  
  Let us first define a RWTA recognizing the subtree series of a finite tree languages.
  
  \begin{definition}
    Let $\Sigma$ be an alphabet.
    Let $L$ be a finite tree language over $\Sigma$.
    The \emph{sequential subtree automaton} associated with $L$ is the RWTA $A_L=(\Sigma,Q,\nu,\delta)$ defined by:
    \begin{itemize}
      \item $Q=\mathrm{SubTreeSet}(L)$,
      \item $\forall t'\in Q$, $\nu(t')=\mathrm{SubTreeSeries}_L(t')$,
      \item $\forall f\in \Sigma$, $\forall t_1,\ldots,t_{k+1}\in Q$, $t_{k+1}\in\delta(f,t_1,\ldots,t_k)$ $\Leftrightarrow$ $t_{k+1}=f(t_1,\ldots,t_k)$.
    \end{itemize}
  \end{definition} 
  
  Notice that by definition, for any tree $t$, $\mathrm{seq}(A_t)$ and $A_{\{t\}}$ are isomorphic.
  
  \begin{example}
    Let us consider the trees $t_1=f(h(a),f(h(a),b))$ and $t_2=f(h(a),h(b))$ defined in Example~\ref{ex subset tree}.
    The RWTA $A_{\{t_1,t_2\}}$ is represented in Figure~\ref{fig ex Al} and realized the series $\mathrm{SubTreeSeries}_{\{t_1,t_2\}}=t+t_2+f(h(a),b)+3h(a)+h(b)+3a+2b$.
  \end{example} 
  
  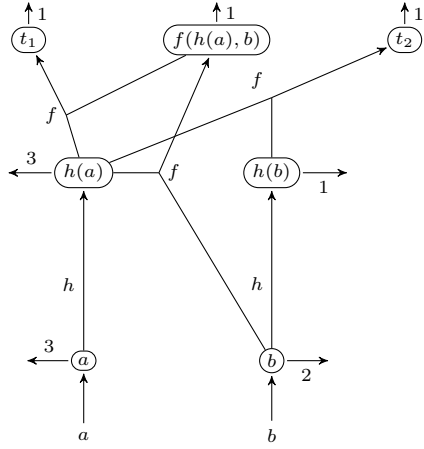
\begin{figure}[H]
    \centerline{
	    \begin{tikzpicture}[node distance=2.5cm,bend angle=30,transform shape,scale=1]
	      \node[state,rounded rectangle] (a) {$a$};
	      \node[state, above of=a,rounded rectangle] (h) {$h(a)$};
	      \node[state, above of=b,rounded rectangle] (hb) {$h(b)$};		
	      \node[state, right of=a,rounded rectangle] (b) {$b$};	   
	      \node[state, above right of=h,rounded rectangle] (f) {$f(h(a),b)$};	
	      \node[state, left of=f,rounded rectangle] (t1) {$t_1$};	 
	      \node[state, right of=f,rounded rectangle] (t2) {$t_2$};	 
        \draw (a) ++(-0.75cm,0cm) edge[<-] node {$3$} (a); 
        \draw (h) ++(-1cm,0cm) edge[<-] node {$3$} (h); 
        \draw (hb) ++(1cm,0cm) edge[<-] node {$1$} (hb); 
        \draw (b) ++(0.75cm,0cm) edge[<-] node {$2$} (b); 
        \draw (t1) ++(0cm,0.5cm) edge[<-] node {$1$} (t1);
        \draw (t2) ++(0cm,0.5cm) edge[<-] node {$1$} (t2);
        \draw (f) ++(0cm,0.5cm) edge[<-] node {$1$} (f);   
        \draw (a) ++(0cm,-1cm) node {$a$}  edge[->] (a);  
        \draw (b) ++(0cm,-1cm) node {$b$}  edge[->] (b);  
        \path[->]
          (a) edge[->,below left] node {$h$} (h)
          (b) edge[->,below left] node {$h$} (hb)
        ;
        \draw (h) ++(1cm,0cm)  edge[->] node[right,pos=0] {$f$} (f) edge[shorten >=0pt] (h) edge[shorten >=0pt] (b);
        \draw (t1) ++(0.5cm,-1cm)  edge[->] node[left,pos=0] {$f$} (t1) edge[shorten >=0pt] (h) edge[shorten >=0pt] (f);
        \draw (hb) ++(0cm,1cm)  edge[->] node[above left,pos=0] {$f$} (t2) edge[shorten >=0pt] (h) edge[shorten >=0pt] (hb);
      \end{tikzpicture}
    }
    \caption{The RWTA $A_{\{t_1,t_2\}}$.}
    \label{fig ex Al}
  \end{figure}
  
  Similarly to the case of the sequential subtree RWTA of a tree, the subtree RWTA $A$ of a language needs the subtree series to be \emph{a priori} known.
  Let us show that $A$ can be computed without knowing the series.
  In order to compute it, we make use of the sum and of the sequentialization, two operations defined in Section~\ref{sec tree ser rwta}. 
  
  \begin{lemma}\label{lem delta r eq r}
    Let $\Sigma$ be an alphabet.
    Let $L$ be a finite tree language over $\Sigma$.
    Let $A_L=(\Sigma,Q,\nu,\delta)$.
    For any tree $r$ in $\mathrm{SubTreeSet}(L)$, it holds:
    
    \centerline{
      $\Delta(r)=\{r\}$.
    }
  \end{lemma}
  \begin{proof}
    By induction over the structure of $r=f(r_1,\ldots,r_k)$.
    By definition of $\Delta$, $\Delta(r)=\delta(f,\Delta(r_1),\ldots,\Delta(r_k))$.
    By induction hypothesis, $\Delta(r_i)=\{r_i\}$.
    Thus, $\Delta(r)=\delta(f,r_1,\ldots,r_k)$.
    Hence $\Delta(r)=\{r\}$.
    \cqfd
  \end{proof}
  
  As a direct consequence of Lemma~\ref{lem delta r eq r}, $A_L$ is sequential.
  Let us show now that this RWTA can be obtained by an inductive sequentialization.
  
  \begin{proposition}\label{prop al iso sum seq}
    Let $\Sigma$ be an alphabet.
    Let $L_1$ and $L_2$ be two distinct finite tree languages over $\Sigma$.
    Then:
    
    \centerline{
      The RWTA $A_{L_1\cup L_2}$ is isomorphic to the accessible part of the sequential RWTA associated with $A_{L_1}+A_{L_2}$.
    }    
  \end{proposition}
  \begin{proof}
    By recurrence over the cardinality of $L_2$.    
    \begin{enumerate}
      \item If $L_2$ is empty, then the proposition is satisfied.
      \item Suppose that $L_2=L'_2\cup\{t\}$ with $t\notin L'_2$.
        Then according to the recurrence hypothesis, the RWTA $A'=A_{L_1\cup L'_2}=(\Sigma,Q',\nu',\delta')$ is isomorphic to the sequential RWTA associated with $A_{L_1}+A_{L'_2}$.
        Let $A_{\{t\}}=(\Sigma,Q_t,\nu_t,\delta_t)$, $A''=A'+A_{\{t\}}=(\Sigma,Q'',\nu'',\delta'')$ and $A'''=(\Sigma,Q''',\nu''',\delta''')$ be the sequential RWTA associated with $A''$.
        By construction, either $Q'\cap \mathrm{SubTreeSet}(t)$ is empty, or the states $r$ of $A_{\{t\}}$ have to be relabelled as $\overline{r}$.
        \begin{enumerate}
          \item By construction, if $Q'\cap \mathrm{SubTreeSet}(t)$ is empty, it holds from Lemma~\ref{lem arbre ht arriv dans t} and from Lemma~\ref{lem delta r eq r} that the construction of the accessible part of the sequential RWTA associated with $A''$ is just a relabelling of the states.
            Furthermore, it can be shown by definition of $A''$ that $A''=A_{L_1\cup L_2}$.
            Hence $A_{L_1\cup L_2}$ is isomorphic to the acecessible part of the sequential RWTA associated with $A_{L_1}+A_{L_2}$.
          \item Otherwise, according to Lemma~\ref{lem delta' eq delta}, $\Delta'''(r)=\{\Delta''(r)\}$, that equals by construction to the set $\{\Delta_t(r)\cup\Delta'(r)\}$.
            Hence the states of the accessible part of $A'''$ are $\mathrm{SubtreeSet}(L_1\cup L'_2)\cup \mathrm{SubtreeSet}(t)$ that equals by definition to $\mathrm{SubtreeSet}(L_1\cup L_2)$.
            Furthermore, for any state $r$ in $Q'''$, 
            
            \centerline{
              $\nu'''(\Delta'''(r))=\nu''(\Delta''(r))=
                \left\{
                  \begin{array}{l@{\ }l}
                    \nu''(\{r,\overline{r}\})=\nu_t(r)+\nu'(r) & \text{ if } r\in Q'\cap \mathrm{SubTreeSet}(t),\\
                    \nu''(\{r\})=\nu'(r) & \text{ if }  r\in Q'\setminus \mathrm{SubTreeSet}(t),\\
                    \nu''(\{\overline{r}\})=\nu_t(\overline{r}) & \text{ if }  r\in \mathrm{SubTreeSet}(t) \setminus Q'.\\
                  \end{array}
                \right.$
            }
            
            Consequently, since $A'''$ is sequential, for any state $r$ in $Q'''$, 
            
            \centerline{
              $\nu'''(r)=\nu'''(\Delta'''(r))=\mathrm{SubTreeSeries}_{L_1\cup L'_2}(r)+\mathrm{SubTreeSeries}_{t}(r)=\mathrm{SubTreeSeries}_{L_1\cup L_2}(r)$.
            }
            
            Finally, since by construction of $A'''$, $\forall f\in \Sigma$, $\forall \{t_1\},\ldots,\{t_{k+1}\}\in Q'''$, $\{t_{k+1}\}\in\delta(f,\{t_1\},\ldots,\{t_k\})$ $\Leftrightarrow$ $t_{k+1}=f(t_1,\ldots,t_k)$, it holds that $A'''$ is isomorphic to $A_{L_1\cup L_2}$.
        \end{enumerate}
    \end{enumerate}   
    \cqfd
  \end{proof} 
  
  Notice that since the complexity of the computation of the accessible part of $A_{L_1}+A_{L_2}$ is equal to the size of $A_{L_1\cup L_2}$, by setting for any tree language $L$, $|L|=\sum_{t\in L}|t|$, it can be performed in time equal to $|L_1|+|L_2|$, and not in an exponential time.  
  Moreover, as a direct consequence of Proposition~\ref{prop seq pres ser}, Proposition~\ref{prop sum aut is sum ser}, Proposition~\ref{prop series real somme lang down} and Proposition~\ref{prop al iso sum seq}
  
  \begin{corollary}\label{cor lang subtree aut ok}
    Let $\Sigma$ be an alphabet.
    Let $L$ be a finite tree language over $\Sigma$.
    Then:
    
    \centerline{
      The RWTA $A_L$ is a sequential RWTA that realizes $\mathrm{SubTreeSeries}_L$.
    }    
  \end{corollary}
  
  Consequently the subtree automaton $A_L$ associated with a finite tree language $L$ can be computed by summing and sequentializing all the $A_{\{t\}} \equiv A_t$ for $t$ in $L$.
  Therefore, since the coomputation of the sequential subtree RWTA of any tree, the sum and the sequentialization (in this case) can be computed in linear time, it holds that, :
  
  \begin{corollary}\label{cor tps cons al}
    Let $\Sigma$ be an alphabet.
    Let $L$ be a finite tree language over $\Sigma$.
    Then:
    
    \centerline{
      The RWTA $A_L$ can be computed in time $O(|L|)$.
    }        
  \end{corollary}
  
  \subsection{Kernel Computation}
  
  In order to compute the subset kernel of two finite tree languages $L_1$ and $L_2$, we first compute the two RWTAs $A_{L_1}$ and $A_{L_2}$;
  Then we compute the cartesian product $A_{L_1}\times A_{L_2}$;
  Finally we sum all the root weight of this RWTA.
  
  Let us first show that our \emph{modus operandi} is correct:
  
  \begin{theorem}\label{thm calc ker subtree}
    Let $\Sigma$ be an alphabet.
    Let $L_1$ and $L_2$ be two finite tree languages over $\Sigma$.
    Let $(\Sigma,Q,\nu,\delta)$ be the accessible part of $A_{L_1}\times A_{L_2}$.
    Then 
    
    \centerline{
      $\mathrm{KerSeries}(L_1,L_2)=\nu(Q)$.
    }
  \end{theorem}
  \begin{proof}    
    Let us set $A_{L_1}=(\Sigma,Q_1,\nu_1,\delta_1)$ and $A_{L_2}=(\Sigma,Q_2,\nu_2,\delta_2)$.
    
    By definition of the series product and from Corollary~\ref{cor lang subtree aut ok}, 
    
    \centerline{
      \begin{tabular}{l@{\ }l}
        $\mathrm{KerSeries}(L_1,L_2)$ & $=\sum_{t\in T_\Sigma}(\mathrm{SubTreeSeries}_{L_1}\times \mathrm{SubTreeSeries}_{L_2})(t)$\\
        & $=\sum_{t\in T_\Sigma}(\mathrm{SubTreeSeries}_{L_1}(t)\times \mathrm{SubTreeSeries}_{L_2}(t))$\\
        & $=\sum_{t\in T_\Sigma}(\nu_1(t)\times \nu_2(t))$.
      \end{tabular}
    }
    
    According to Lemma~\ref{lem delta r eq r}, for any tree $t$ in $T_\Sigma$, $\Delta_1(t)$ (resp. $\Delta_2(t)$) is either equal to $\{t\}$ if $t\in\mathrm{SubTreeSet}(L_1)$ ($t\in\mathrm{SubTreeSet}(L_2)$) or to $\emptyset$.
    Therefore, according to Lemma~\ref{lem delta pour prod cart}, $\Delta(t)$ is either equal to $\{(t,t)\}$ if $t\in\mathrm{SubTreeSet}(L_1)\cap \mathrm{SubTreeSet}(L_2)$, $\emptyset$ otherwise.
    Moreover, by definition of $\nu$, for any tree $t$, $\nu(t)=\nu(\Delta(t))=\nu_1(t)\times\nu_2(t)$.
    Furthermore, by definition of $Q$, $t$ is in $Q$ if and only if $t\in\mathrm{SubTreeSet}(L_1)\cap \mathrm{SubTreeSet}(L_2)$.
    Consequently, $\nu(Q)=\sum_{t\in\mathrm{SubTreeSet}(L_1)\cap \mathrm{SubTreeSet}(L_2)} \nu_1(t)\times\nu_2(t)$.
    Since for any tree $t$, if $t\notin \mathrm{SubTreeSet}(L_1)$ (resp. $t\notin \mathrm{SubTreeSet}(L_2)$), then $\nu_1(t)=0$ (resp. $\nu_2(t)=0$), it holds that
    
    \centerline{
      $\nu(Q)=\sum_{t\in T_\Sigma} \nu_1(t)\times\nu_2(t)=\mathrm{KerSeries}(L_1,L_2)$.
    }
    \cqfd
  \end{proof}
  
  Finally, by combining the elemental complexities,
  
  \begin{theorem}
    Let $\Sigma$ be an alphabet.
    Let $L_1$ and $L_2$ be two finite tree languages over $\Sigma$.
    Then 
    
    \centerline{
      $\mathrm{KerSeries}(L_1,L_2)$ can be computed in time $O(|L_1|+|L_2|+\mathrm{Card}(\mathrm{SubTreeSet}(L_1)\cap \mathrm{SubTreeSet}(L_2)))$.
    }    
  \end{theorem}
  \begin{proof}
    From Corollary~\ref{cor tps cons al}, $A_{L_1}=(\Sigma,Q_1,\nu_1,\delta_1)$ and $A_{L_2}=(\Sigma,Q_2,\nu_2,\delta_2)$ are constructed in time $O(|L_1|+|L_2|)$.
    From Lemma~\ref{lem delta pour prod cart}, the accessible part of $A_{L_1}\times A_{L_2}$ is composed of the set $S$ of states of the form $(t,t)$ with $t\in \mathrm{SubTreeSet}(L_1)\cap \mathrm{SubTreeSet}(L_2)$.  
    From Theorem~\ref{thm calc ker subtree}, $\mathrm{KerSeries}(L_1,L_2)$ is computed summing the root weight of the states in $S$.
    Therefore $\mathrm{KerSeries}(L_1,L_2)$ can be computed in time $O(|L_1|+|L_2|+\mathrm{Card}(\mathrm{SubTreeSet}(L_1)\cap \mathrm{SubTreeSet}(L_2)))$.
    \cqfd
  \end{proof} 
    
  \begin{example}
    Let us consider the trees $t_1=f(h(a),f(h(a),b))$, $t_2=f(h(a),h(b))$ and $t_3=f(f(b,h(b)),f(h(a),h(b)))$ defined in Example~\ref{ex subset tree}.
    The RWTA $A_{\{t_3\}}$ is represented in Figure~\ref{fig ex At3}.
    The RWTA $R=A_{\{t_1,t_2\}}\times A_{\{t_3\}}$ is represented in Figure~\ref{fig ex Alprod}.
    The sum of the root weights of $R$ is equal to $15$, that is $\mathrm{KerSeries}(\{t_1,t_2\},\{t_3\})$.
  \end{example} 
  
  \begin{minipage}{0.45\linewidth}
	  \begin{figure}[H]
	    \centerline{
		    \begin{tikzpicture}[node distance=2.5cm,bend angle=30,transform shape,scale=1]
		      \node[state,rounded rectangle] (a) {$a$};
		      \node[state, above of=a,rounded rectangle] (h) {$h(a)$};
		      \node[state, above of=b,rounded rectangle] (hb) {$h(b)$};		
		      \node[state, right of=a,rounded rectangle] (b) {$b$};	   
		      \node[state, above of=hb,rounded rectangle] (t2) {$f(h(a),h(b))$};
		      \node[state, above left of=hb,rounded rectangle] (f) {$f(b,h(b))$};	
		      \node[state, above left of=f,rounded rectangle] (t3) {$t_3$};	 
	        \draw (a) ++(-0.75cm,0cm) edge[<-] node {$1$} (a); 
	        \draw (h) ++(-1cm,0cm) edge[<-] node {$1$} (h); 
	        \draw (hb) ++(1cm,0cm) edge[<-] node {$2$} (hb); 
	        \draw (b) ++(0.75cm,0cm) edge[<-] node {$3$} (b); 
	        \draw (t2) ++(0cm,0.5cm) edge[<-] node {$1$} (t2);
	        \draw (f) ++(-1cm,0cm) edge[<-] node {$1$} (f);
	        \draw (t3) ++(0cm,0.5cm) edge[<-] node {$1$} (t3);  
	        \draw (a) ++(0cm,-1cm) node {$a$}  edge[->] (a);  
	        \draw (b) ++(0cm,-1cm) node {$b$}  edge[->] (b);  
	        \path[->]
	          (a) edge[->,below left] node {$h$} (h)
	          (b) edge[->,below left] node {$h$} (hb)
	        ;
	        \draw (hb) ++(0cm,1cm)  edge[->] node[above left,pos=0] {$f$} (t2) edge[shorten >=0pt] (h) edge[shorten >=0pt] (hb);
	        \draw (hb) ++(-1cm,-1cm)  edge[->] node[below left,pos=0] {$f$} (f) edge[shorten >=0pt] (hb) edge[shorten >=0pt] (b);
	        \draw (f) ++(-0cm,1cm)  edge[->] node[above right,pos=0] {$f$} (t3) edge[shorten >=0pt] (f) edge[shorten >=0pt] (t2);
	      \end{tikzpicture}
	    }
	    \caption{The RWTA $A_{\{t_3\}}$.}
	    \label{fig ex At3}
	  \end{figure}
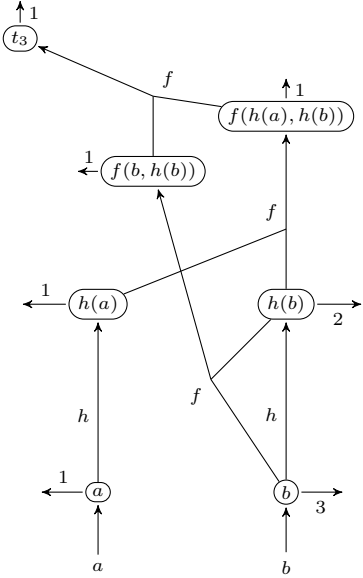
	\end{minipage}
	\hfill
	\begin{minipage}{0.45\linewidth}  
	  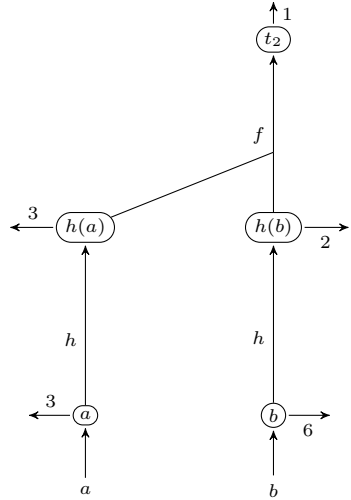
\begin{figure}[H]
	    \centerline{
		    \begin{tikzpicture}[node distance=2.5cm,bend angle=30,transform shape,scale=1]
		      \node[state,rounded rectangle] (a) {$a$};
		      \node[state, above of=a,rounded rectangle] (h) {$h(a)$};
		      \node[state, above of=b,rounded rectangle] (hb) {$h(b)$};		
		      \node[state, right of=a,rounded rectangle] (b) {$b$};	   
		      \node[state, above of=hb,rounded rectangle] (t2) {$t_2$};	 
	        \draw (a) ++(-0.75cm,0cm) edge[<-] node {$3$} (a); 
	        \draw (h) ++(-1cm,0cm) edge[<-] node {$3$} (h); 
	        \draw (hb) ++(1cm,0cm) edge[<-] node {$2$} (hb); 
	        \draw (b) ++(0.75cm,0cm) edge[<-] node {$6$} (b); ;
	        \draw (t2) ++(0cm,0.5cm) edge[<-] node {$1$} (t2);  
	        \draw (a) ++(0cm,-1cm) node {$a$}  edge[->] (a);  
	        \draw (b) ++(0cm,-1cm) node {$b$}  edge[->] (b);  
	        \path[->]
	          (a) edge[->,below left] node {$h$} (h)
	          (b) edge[->,below left] node {$h$} (hb)
	        ;
	        \draw (hb) ++(0cm,1cm)  edge[->] node[above left,pos=0] {$f$} (t2) edge[shorten >=0pt] (h) edge[shorten >=0pt] (hb);
	      \end{tikzpicture}
	    }
	    \caption{The RWTA $A_{\{t_1,t_2\}}\times A_{\{t_3\}}$.}
	    \label{fig ex Alprod}
	  \end{figure}
	\end{minipage}
	
	\section{Conclusion and Perspectives}
	
	  In this paper, we defined new weighted tree automata that are always sequentializable but that does not realize all the classical recognizable series.
	  We studied the different algebraic combinations of these automata (sum, products, regular operations) in order to determine their closures.
	  Once these definitions stated, we made use of these new structures in order to compute the subtree kernel of two finite tree series in an efficient way.
	  
	  Our technique can be applied to other computations. 
	  Indeed, other tree kernels exist, like the SST kernel. 
	  The next step of our work is to apply our constructions in order to efficiently compute these kernels.
	  However, this application is not so direct since it seems that the SST series may not be sequentializable w.r.t. a linear space complexity.
	  Hence we have to find different techniques, like extension of lookahead determinism~\cite{HW08} for example.
	  
	  Another perspective is related to the series realized by RWTAs. 
	  It is an open question to determine what family they exactly are.

	\bibliography{biblio}
\end{document}